\documentclass[letterpaper, 10 pt, conference]{ieeeconf}  
\IEEEoverridecommandlockouts                              
\usepackage[utf8]{inputenc}
\usepackage{amsmath}
\usepackage{amsthm}
\usepackage{mathrsfs}
\usepackage{amssymb}
\usepackage{color}
\usepackage{cite} 
\usepackage[font=footnotesize]{caption}
\usepackage{comment}
\usepackage{float}
\usepackage{graphicx}
\usepackage{float}
\usepackage{lscape}

\usepackage{bm}
\usepackage{booktabs}
\usepackage{longtable,tabularx,ragged2e}
\usepackage[super]{nth}
\usepackage{multicol}
\usepackage{booktabs}
\usepackage{lscape}
\usepackage{nicefrac}
\usepackage[hidelinks]{hyperref}
\usepackage{mathtools}
\usepackage{siunitx}
\usepackage{algorithm}
\usepackage[noend]{algpseudocode}

\newtheorem{lma}{Lemma}
\newtheorem{theo}{Theorem}
\newtheorem{prop}{Proposition}

\newtheorem{defn}{Definition}

\DeclareMathOperator{\rank}{rank}

\newcommand{\isdef}{\triangleq}

\newcommand{\eigmax}{\boldsymbol{\lambda_{\rm max}}}
\newcommand{\eigmin}{\boldsymbol{\lambda_{\rm min}}}

\newcommand{\rmT}{{\rm T}}

\newcommand{\BBN}{{\mathbb N}}

\newcommand{\BBR}{{\mathbb R}}

\algrenewcommand\algorithmicrequire{\textbf{Initialize:}}

\newcounter{example}
\newenvironment{example}[1][]{\refstepcounter{example}\par\medskip
   \noindent \textbf{\indent Example~\theexample. #1} \rmfamily}{\medskip}
   
\DeclareMathOperator*{\argmin}{arg\,min}
\newcommand{\Mod}[1]{\, \mathrm{mod}\, #1}

\newtheorem{lemalt}{Lemma}[lma]

\newenvironment{lema}[1]{
  \renewcommand\thelemalt{#1}
  \lemalt
}{\endlemalt}

\title{Recursive Least Squares with Fading Regularization\\ for Finite-Time Convergence without Persistent Excitation}

\author{\large Brian Lai, Dimitra Panagou, and Dennis S. Bernstein 
\thanks{ Brian Lai, Dimitra Panagou, and Dennis S. Bernstein are with the Department of Aerospace Engineering, University of Michigan, Ann Arbor, MI, USA. {\tt\small \{brianlai, dpanagou, dsbaero\}@umich.edu}}
}

\begin{document}
\maketitle

\begin{abstract}
This paper extends recursive least squares (RLS) to include time-varying regularization.
%
%
%
This extension provides flexibility for updating the least squares regularization term in real time.
Existing results with constant regularization
%
imply that the parameter-estimation error dynamics of RLS are globally attractive to zero if and only the regressor is weakly persistently exciting.
This work shows that, by extending classical RLS to include a time-varying (fading) regularization term that converges to zero, the parameter-estimation error dynamics are globally attractive to zero without weakly persistent excitation.
Moreover, if the fading regularization term converges to zero in finite time, then the parameter estimation error also converges to zero in finite time.
Finally, we propose rank-1 fading regularization (R1FR) RLS, a time-varying regularization algorithm with fading regularization that converges to zero, and which runs in the same computational complexity as classical RLS.
%
%
%
Numerical examples are presented to validate theoretical guarantees and to show how R1FR-RLS can protect against over-regularization.
\end{abstract}


\section{Introduction}

Classical recursive least squares (RLS) regularizes the parameter estimate with a fixed weighting selected by the user \cite{islam2019recursive}.
Regularization provides numerical stability during startup and reduces the variance of the estimator, at the cost of a bias in the estimate \cite{lai2021regularization}, known as the bias-variance tradeoff within the context of batch least squares \cite[pp. 223-224]{hastie2009elements}.
%
%
However, as more data are collected, it may be desirable to update the regularization weighting based on new external knowledge. 
%
%
%
It may also be desirable to remove regularization entirely in order to remove the regularization-induced bias. 
This can be accomplished using time-varying regularization.

Versions of RLS with time-varying regularization are considered in \cite{ali2016stability,mahadi2022recursive}.
The method in \cite{mahadi2022recursive} uses a recursion to minimize the least squares cost with time-varying regularization.
This method uses a first-order Taylor series approximation of the covariance update, which results in an approximate solution. 
%
%
%
Moreover, computing the resulting quadratic term of the expansion requires $\mathcal{O}(n^3)$ complexity, whereas classical RLS is run in $\mathcal{O}(pn^2)$ complexity, where $n$ is the number of parameters, $p$ is the number of measurements per step, and typically $n \gg p$.
A higher order expansion may achieve better approximation at the expense of computational complexity.
%
%
The method developed in \cite{ali2016stability} uses  time-varying regularization with a sliding window of data,  that is, a fixed number of  data points.
This method has the same computational complexity as classical RLS, and conditions are given for  convergence of parameter estimation error to zero.
However, a major assumption the analysis of \cite{ali2016stability} is that the data are persistently exciting.

Persistent excitation is a key condition in stability analysis of parameter-estimation methods, and significant effort has been put toward weakening the PE condition \cite{parikh2019integral, panteley2001relaxed, barabanov2017global, aranovskiy2016parameters, bruce2021necessary}.
In particular, \cite{bruce2021necessary} presents \textit{weakly persistent excitation} as a necessary and sufficient condition for global attractivity of parameter estimation error to zero in RLS .
This present paper shows, however,   that  weakly persistent excitation is only necessary for overcoming of effect of regularization-induced bias.

This work shows that, by extending classical RLS to include a time-varying regularization term that converges to zero,  the parameter-estimation error dynamics are globally attractive to zero without weakly persistent excitation.
In particular, it is sufficient that the data collected attains full rank.
Moreover, if the regularization term is designed to converge to zero in finite time, then the parameter estimation error also converges to zero in finite time.
Finally, we show that, if the difference between regularization matrices in successive steps is of rank 1, then RLS with time-varying regularization can be implemented with the same computational complexity as classical RLS.
We present an algorithm, called \textit{rank-1 fading-regularization RLS} (R1FR-RLS), which satisfies the aforementioned conditions.

\section{Batch Least Squares with Time-Varying Regularization}
For all $k \ge 0$, let $y_k \in \BBR^p$ be the measurement, let $\phi_k \in \BBR^{p \times n}$ be the regressor, and let positive-definite $\Gamma_k \in \BBR^{p \times p}$ be the weighting matrix.
Moreover, for all $k \ge 0$, let positive-definite $R_k \in \BBR^{n \times n}$ be the regularization matrix and let $\theta_{{\rm reg},k}$ be the a regularization parameter.
For all $k \ge 0$, define the least squares cost with time-varying regularization, $J_{k} \colon \BBR^{n} \rightarrow \BBR$, as
\begin{align}
\label{eqn: Jk definition}
     J_{k}(\hat\theta) &\triangleq \sum_{i=0}^{k}( y_i -  \phi_i \hat\theta )^\rmT \Gamma_i ( y_i - \phi_i \hat\theta) \nonumber
     \\
     & \hspace{50pt} + (\hat\theta- \theta_{{\rm reg},k})^\rmT  R_k (\hat\theta- \theta_{{\rm reg},k}).
\end{align}
For all $k \ge 0$, denote the global minimizer of the function $J_k$, defined in \eqref{eqn: Jk definition}, by 
\begin{align}
    \theta_{k+1} \isdef \underset{ \hat\theta \in \BBR^n  }{\operatorname{argmin}} \ J_{k}(\hat\theta).\label{eqn: Jk minimizer}
\end{align}
It follows that, for all $k \ge 0$, the batch least squares minimizer of \eqref{eqn: Jk definition} is given by
\begin{align}
    \theta_{k+1} = (R_k + \sum_{i=0}^{k} \phi_i^\rmT \Gamma_i \phi_i)^{-1}( R_k \theta_{{\rm reg},k} + \sum_{i=0}^k \phi_i^\rmT \Gamma_i y_i).
    \label{eqn: time-varying regularization batch least squares}
\end{align}
Note that since the least squares cost \eqref{eqn: Jk definition} is minimized in batch, the classical least squares solution \eqref{eqn: time-varying regularization batch least squares} can be used.
We still present batch solution \eqref{eqn: time-varying regularization batch least squares} however, as it will be useful for later analysis.

\section{Recursive Least Squares with Time-Varying Regularization}
Theorem \ref{theo: time-varying regularization RLS} provides a statement and proof of RLS with time-varying regularization.
This algorithm recursively computes the minimizer of the cost function \eqref{eqn: Jk definition}.
\begin{theo}
\label{theo: time-varying regularization RLS}
For all $k\in \BBN_0$, let $\phi_k\in \BBR^{p\times n}$, let $y_k \in\BBR^p$, and let $\Gamma_k \in \BBR^{p \times p}$ be positive semidefinite.
Next, for all $k \ge 0$, let $\theta_{{\rm reg},k} \in \BBR^{n}$ and let $R_k \in \BBR^{n\times n}$ be positive semidefinite and satisfy
\begin{align}
\label{eqn: variable regularization pos def condition}
    R_k + \sum_{i=0}^k \phi_i^\rmT \Gamma_i \phi_i \succ 0.
\end{align}
Then, for all $k \ge 0$, $J_{k} \colon \BBR^{n} \rightarrow \BBR$, defined in \eqref{eqn: Jk definition}, has a unique global minimizer, denoted as \eqref{eqn: Jk minimizer}.
Furthermore, $\theta_1$ is given by
\begin{align}
    P_1^{-1} &= R_0 + \phi_0^\rmT \Gamma_0 \phi_0,
    \label{eqn: variable regularization P1 update}
    \\
    \theta_1 &= \theta_{{\rm reg},0} + P_1 \phi_0^\rmT \Gamma_0( y_0 - \phi_0 \theta_{{\rm reg},0}),
    \label{eqn: variable regularization theta1 update}
\end{align}
and, for all $k \ge 1$, $\theta_{k+1}$ is given by
\begin{align}
     P_{k+1}^{-1} &= P_{k}^{-1} + \phi_k^\rmT \Gamma_k \phi_k + R_k - R_{k-1},
     \label{eqn: variable regularization Pk update}
    \\
    \theta_{k+1} &= \theta_k + P_{k+1}[\phi_k^\rmT \Gamma_k (y_k - \phi_k \theta_k) + R_k(\theta_{{\rm reg},k} - \theta_k) \nonumber
    \\
    & \hspace{70pt} - R_{k-1}(\theta_{{\rm reg},k-1} - \theta_k)].
    \label{eqn: variable regularization thetak update}
\end{align}
\end{theo}

\begin{proof}
To begin, we show that \eqref{eqn: variable regularization P1 update} and \eqref{eqn: variable regularization theta1 update} hold. Note that $J_0 \colon \BBR^n \rightarrow \BBR$ can be written as
\begin{align*}
    J_{0}(\hat\theta) = \hat{\theta}^\rmT A_0 \hat{\theta} + 2 b_0^\rmT \hat{\theta} + c_0,
\end{align*}
where
\begin{align*}
    A_0 &\triangleq \phi_0^\rmT \Gamma_0 \phi_0 + R_0, 
    \\
    b_0 &\triangleq - \phi_0^\rmT \Gamma_0 y_0 - R_0 \theta_{{\rm reg},0},
    \\
    c_0 &\triangleq y_0^\rmT \Gamma_0 y_0 + \theta_{{\rm reg},0}^\rmT R_0 \theta_{{\rm reg},0}.
\end{align*}
It follows from \eqref{eqn: variable regularization pos def condition} that $A_0$ is positive definite, thus nonsingular. 
Defining $P_1 \triangleq A_0^{-1}$ yields that
\begin{align*}
    P_1^{-1} &= \phi_0^\rmT \Gamma_0 \phi_0 + R_0.
\end{align*}
Hence, \eqref{eqn: variable regularization P1 update} is satisfied.
Furthermore, it follows from Lemma \ref{lem: quadratic cost minimizer} that $J_0$ has a unique global minimizer $\theta_1$ given by
\begin{align*}
    \theta_1 &= -A_0^{-1} b_0
    \\
    &= P_1(\phi_0^\rmT \Gamma_0 y_0 + R_0 \theta_{{\rm reg},0})
    \\
    &= P_1[\phi_0^\rmT \Gamma_0 y_0 + (P_1^{-1} - \phi_0^\rmT \Gamma_0 \phi_0) \theta_{{\rm reg},0}]
    \\
    &= \theta_{{\rm reg},0} + P_1(\phi_0^\rmT \Gamma_0 y_0 -  \phi_0^\rmT \Gamma_0 \phi_0 \theta_{{\rm reg},0}) 
    \\
    &= \theta_{{\rm reg},0} + P_1 \phi_0^\rmT \Gamma_0 ( y_0 - \phi_0 \theta_{{\rm reg},0} ).
\end{align*}
Hence, \eqref{eqn: variable regularization theta1 update} is satisfied.

Next, we show that \eqref{eqn: variable regularization Pk update} and \eqref{eqn: variable regularization thetak update} hold by induction on $k \ge 1$.
Note that $J_1 \colon \BBR^n \rightarrow \BBR$ can be written as
\begin{align*}
    J_{1}(\hat\theta) = \hat{\theta}^\rmT A_1 \hat{\theta} + 2 b_1^\rmT \hat{\theta} + c_1
\end{align*}
where
\begin{align*}
    A_1 &\triangleq \phi_0^\rmT \Gamma_0 \phi_0 + \phi_1^\rmT \Gamma_1 \phi_1 + R_1, 
    \\
    b_1 &\triangleq - \phi_0^\rmT \Gamma_0 y_0 - \phi_1^\rmT \Gamma_1 y_1 - R_1 \theta_{{\rm reg},1},
    \\
    c_1 &\triangleq y_0^\rmT \Gamma_0 y_0 + y_1^\rmT \Gamma_1 y_1 + \theta_{{\rm reg},1}^\rmT R_1 \theta_{{\rm reg},1}.
\end{align*}
It follows from \eqref{eqn: variable regularization pos def condition} that $A_1$ is positive definite, thus nonsingular. Furthermore, substituting in the definitions of $A_0$ and $b_0$, it follows that
\begin{align*}
    A_1 
    &= A_0 +  \phi_1^\rmT \Gamma_1 \phi_1 + R_1 - R_0 .
    \\
    b_1 
    &= b_{0} - \phi_1^\rmT \Gamma_1 y_1 - R_1 \theta_{{\rm reg},1} + R_0 \theta_{{\rm reg},0}.
\end{align*}
Defining $P_2 \triangleq A_1^{-1}$, it follows that \eqref{eqn: variable regularization Pk update} holds for $k=1$.
Next, it follows from Lemma \ref{lem: quadratic cost minimizer} that $J_1$ has a unique global minimizer $\theta_2$ given by
\begin{align*}
    \theta_{2} &= -A_1^{-1} b_1 
    \\
    &= A_1^{-1}[\phi_1^\rmT \Gamma_1 y_1 + R_1 \theta_{{\rm reg},1} - R_0 \theta_{{\rm reg},0} - b_{0}]
    \\
    &= A_1^{-1}[\phi_1^\rmT \Gamma_1 y_1 + R_1 \theta_{{\rm reg},1} - R_0 \theta_{{\rm reg},0} + A_{0} \theta_1]
    \\
    &= A_1^{-1}[\phi_1^\rmT \Gamma_1 y_1 + R_1 \theta_{{\rm reg},1} - R_0 \theta_{{\rm reg},0}
    \\
    & \hspace{90pt} + (A_{1} - \phi_1^\rmT \Gamma_1 \phi_1 - R_1 + R_{0}) \theta_1]
    \\
    &= \theta_1 + A_1^{-1}[\phi_1^\rmT \Gamma_1 y_1 + R_1 \theta_{{\rm reg},1} - R_0 \theta_{{\rm reg},0}
    \\
    & \hspace{90pt} + (- \phi_1^\rmT \Gamma_1 \phi_1 - R_1 + R_{0}) \theta_1]
    \\
    &= \theta_k + P_2[\phi_1^\rmT \Gamma_1 (y_1 - \phi_1 \theta_1) + R_1(\theta_{{\rm reg},1} - \theta_1) 
    \\
    & \hspace{90pt} - R_0(\theta_{{\rm reg},0} - \theta_1)].
\end{align*}
Hence, \eqref{eqn: variable regularization thetak update} is satisfied for $k = 1$.

Now, let $k \ge 2$. Then, $J_k \colon \BBR^n \rightarrow \BBR$ can be written as
\begin{align*}
    J_{k}(\hat\theta) = \hat{\theta}^\rmT A_k \hat{\theta} + 2 b_k^\rmT \hat{\theta} + c_k
\end{align*}
where
\begin{align*}
    A_k &\triangleq \sum_{i=0}^{k} \phi_i^\rmT \Gamma_i \phi_i + R_k, 
    \\
    b_k &\triangleq - \sum_{i=0}^k \phi_i^\rmT \Gamma_i y_i - R_k \theta_{{\rm reg},k}
    \\
    c_k &\triangleq \sum_{i=0}^k y_i^\rmT \Gamma_i y_i + \theta_{{\rm reg},k}^\rmT R_k \theta_{{\rm reg},k}
\end{align*}
Furthermore, $A_k$ and $b_k$ can be written recursively as 
\begin{align*}
    A_{k} &= A_{k-1} + \phi_k^\rmT \Gamma_k \phi_k + R_k - R_{k-1},
    \\
    b_k &= b_{k-1} - \phi_k^\rmT \Gamma_k  y_k - R_k \theta_{{\rm reg},k} + R_{k-1} \theta_{{\rm reg},k-1}.
\end{align*}
It follows from \eqref{eqn: variable regularization pos def condition} that $A_k$ is positive definite, thus nonsingular. Defining $P_{k+1} \triangleq A_k^{-1}$, it follows that $\eqref{eqn: variable regularization Pk update}$ is satisfied.
Furthermore, it follows from Lemma \ref{lem: quadratic cost minimizer} that $J_k$ has a unique global minimizer $\theta_{k+1}$ given by
\begin{align*}
    & \theta_{k+1} = -A_k^{-1} b_k 
    \\
    &= A_k^{-1}[\phi_k^\rmT \Gamma_k y_k + R_k \theta_{{\rm reg},k} - R_{k-1} \theta_{{\rm reg},k-1} -b_{k-1}]
    \\
    &= A_k^{-1}[\phi_k^\rmT \Gamma_k y_k + R_k \theta_{{\rm reg},k} - R_{k-1} \theta_{{\rm reg},k-1} + A_{k-1} \theta_k]
    \\
    &= A_k^{-1}[\phi_k^\rmT \Gamma_k y_k + R_k \theta_{{\rm reg},k} - R_{k-1} \theta_{{\rm reg},k-1} 
    \\
    & \hspace{70pt} + (A_{k} - \phi_k^\rmT \Gamma_k  \phi_k - R_k + R_{k-1}) \theta_k]
    \\
    &= \theta_k + A_k^{-1}[\phi_k^\rmT \Gamma_k y_k + + R_k \theta_{{\rm reg},k} - R_{k-1} \theta_{{\rm reg},k-1}
    \\
    & \hspace{70pt} + (- \phi_k^\rmT \Gamma_k \phi_k - R_k + R_{k-1}) \theta_k]
    \\
    &= \theta_k + P_{k+1}[\phi_k^\rmT \Gamma_k (y_k - \phi_k \theta_k) + R_k (\theta_{{\rm reg},k} - \theta_k) 
    \\
    & \hspace{70pt} - R_{k-1} (\theta_{{\rm reg},k-1} - \theta_k)].
\end{align*}
Hence, \eqref{eqn: variable regularization thetak update} is satisfied.
\end{proof}

\section{Global Attractivity}

For the analysis of this section, we make the assumption that there exists $\theta \in \BBR^n$ such that, for all $k \ge 0$, 
\begin{align}
    \label{eqn: linear measurement process no noise}
    y_k = \phi_k \theta.
\end{align}
This assumption of a linear model without noise is a commonly made assumption in analysis of RLS \cite{bruce2021necessary,barabanov2017global,lai2024generalized} and will aid in comparing classical RLS to time-varying regularization RLS.
Next, for all $k \ge 1$, define the parameter estimation error $\tilde{\theta}_k \in \BBR^n$ as
\begin{align}
\label{eqn: parameter estimaton error defn}
    \tilde{\theta}_k \triangleq \theta_k - \theta.
\end{align}
Proposition \ref{prop: time-varying regularization error dynamics} gives an expression for the error dynamics of time-varying regularization RLS.

\begin{prop}
\label{prop: time-varying regularization error dynamics}
    Consider the assumptions and notation of Theorem \ref{theo: time-varying regularization RLS}. 
    Assume there exists $\theta \in \BBR^n$ such that, for all $k \ge 0$, \eqref{eqn: linear measurement process no noise} holds.
    Finally, for all $k \ge 1$, define $\tilde{\theta}_k \in \BBR^n$ by \eqref{eqn: parameter estimaton error defn}.
    Then, for all $k \ge 1$, $\tilde{\theta}_{k+1}$ can be expressed as
    \begin{align}
    \label{eqn: time-varying regularization error dynamics}
        \tilde{\theta}_{k+1} = P_{k+1}[ P_k^{-1} \tilde{\theta}_k & + R_k(\theta_{{\rm reg},k} - \theta) \nonumber
        \\
        & - R_{k-1}(\theta_{{\rm reg},k-1} - \theta)]
    \end{align}
\end{prop}

\begin{proof}
    Let $k \ge 1$. Substituting \eqref{eqn: linear measurement process no noise} into \eqref{eqn: variable regularization thetak update}, it follows that
    \begin{align*}
        \theta_{k+1} &= \theta_k + P_{k+1}[\phi_k^\rmT \Gamma_k \phi_k (\theta - \theta_k) + R_k(\theta_{{\rm reg},k} - \theta_k) \nonumber
        \\
        & \hspace{70pt} - R_{k-1}(\theta_{{\rm reg},k-1} - \theta_k)].
    \end{align*}
    Subtracting $\theta$ from both sides and noting that $\theta_{{\rm reg},k} - \theta_k = (-\tilde{\theta}_k + \theta_{{\rm reg},k} - \theta)$ and $\theta_{{\rm reg},k-1} - \theta_k = (-\tilde{\theta}_k + \theta_{{\rm reg},k-1} - \theta)$ yields that
    \begin{align*}
        \tilde{\theta}_{k+1} &= \tilde{\theta}_k + P_{k+1}[- \phi_k^\rmT \Gamma_k \phi_k \tilde{\theta}_k + R_k(-\tilde{\theta}_k + \theta_{{\rm reg},k} - \theta)
        \\
        & \hspace{80pt} - R_{k-1}(-\tilde{\theta}_k + \theta_{{\rm reg},k-1} - \theta)].
    \end{align*}
    Rearranging term, it follows that
    \begin{align}
    \label{eqn: time-varying regularization error dynamics temp 1}
        \tilde{\theta}_{k+1} &= [I_n - P_{k+1}(\phi_k^\rmT \Gamma_k \phi_k + R_k - R_{k-1})] \tilde{\theta}_k
        \nonumber
        \\
        & + P_{k+1} [R_k(\theta_{{\rm reg},k} - \theta) - R_{k-1}(\theta_{{\rm reg},k-1} - \theta)].
    \end{align}
    Next, pre-multiplying both sides of \eqref{eqn: variable regularization Pk update} by $P_{k+1}$ and rearranging terms yields that
    \begin{align}
    \label{eqn: time-varying regularization error dynamics temp 2}
        P_{k+1}P_k^{-1} = I_n - P_{k+1}(\phi_k^\rmT \Gamma_k \phi_k + R_k - R_{k-1}).
    \end{align}
    Finally, substituting \eqref{eqn: time-varying regularization error dynamics temp 2} into \eqref{eqn: time-varying regularization error dynamics temp 1} yields \eqref{eqn: time-varying regularization error dynamics}.
\end{proof}

In RLS analysis, it is desirable to show that the parameter estimation error converges to zero. 
This is often done by showing that the zero equilibrium of the estimation error dynamics is globally attractive \cite{bruce2021necessary,lai2024generalized}.
However, the error dynamics of RLS with time-varying regularization, \eqref{eqn: time-varying regularization error dynamics}, do not have an equilibrium at zero. 
For this reason, Definition \eqref{defn: globally attractive to a point} defines global attractivity to a point, without assuming that point is an equilibrium of the system.
We also define global finite-time attractivity to a point, which indicates that convergence happens in finite time.

\begin{defn}
\label{defn: globally attractive to a point}
    Let $f\colon\BBN_0 \times \BBR^n\to\BBR^n$ and consider the system
    \begin{align}
        x_{k+1} = f(k,x_k),
        \label{eq:NonLinSys}
    \end{align}
    where, for all $k \ge 0$, $x_k \in \BBR^n$. Let $x \in \BBR^n$. The system \eqref{eq:NonLinSys} is \textbf{globally attractive} to $x$ if, for all $k_0  \ge 0$ and $x_{k_0} \in \BBR^n$, $\lim_{k \to \infty} x_k = x$.
    Furthermore, if , for all $k_0  \ge 0$ and $x_{k_0} \in \BBR^n$, there exists $N$ such that, for all $k \ge N$, $x_k = x$, then the system \eqref{eq:NonLinSys} is \textbf{globally finite-time attractive} to $x$.
    
    Finally, we say \eqref{eq:NonLinSys} is globally (finite-time) attractive to zero if it is globally (finite-time) attractive to $0_{n \times 1}$.
\end{defn}

Note that if there exists positive definite $R \in \BBR^{n \times n}$ and $\theta_0 \in \BBR^n$ such that, for all $k \ge 0$, $R_k = R$ and $\theta_{{\rm reg},k} = \theta_0$, then \eqref{eqn: time-varying regularization error dynamics} simplifies to the error dynamics of classical RLS.
Theorem \ref{theo: WPE result from Bruce 2021} gives a necessary and sufficient condition, \eqref{eqn: WPE condition}, for the global attractivity to zero of the classical RLS error dynamics.
The condition \eqref{eqn: WPE condition} is called \textit{weakly persistent excitation} in \cite{bruce2021necessary}, as it is a weaker condition than the classical \textit{persistent excitation} condition (e.g. Definition 3 of \cite{bruce2021necessary}).

\begin{theo}
\label{theo: WPE result from Bruce 2021}
    Consider the assumptions and notation of Proposition \ref{prop: time-varying regularization error dynamics}.
    If there exists positive definite $R \in \BBR^{n \times n}$ and $\theta_0 \in \BBR^n$ such that, for all $k \ge 0$, $R_k = R$ and $\theta_{{\rm reg},k} = \theta_0$, then the \eqref{eqn: time-varying regularization error dynamics} is globally attractive to zero if and only if 
    \begin{align}
    \label{eqn: WPE condition}
        \lim_{k \rightarrow \infty} \eigmin\left[ \sum_{i=0}^k \phi_i^\rmT \Gamma_i \phi_i \right] = \infty.
    \end{align}
\end{theo}
\begin{proof}
    This proof uses results from \cite{bruce2021necessary}, which only considers the case where, for all $k \ge 0$, $\Gamma_k = I_p$. 
    The results of \cite{bruce2021necessary} can easily be extended to the case of positive-definite $\Gamma_k$.
    
    If follows from Theorem 3 of \cite{bruce2021necessary} that the zero equilibrium of \eqref{eqn: time-varying regularization error dynamics} is globally asymptotically stable (Lyapunov stable and globally attractive) if and only if \eqref{eqn: WPE condition} holds. 
    Next, it follows from Proposition 1 of \cite{bruce2021necessary} that  that the zero equilibrium of \eqref{eqn: time-varying regularization error dynamics} is Lyapunov stable, without any assumptions on the regressor $(\phi_k)_{k=0}^\infty$. 
    Hence, the zero equilibrium of \eqref{eqn: time-varying regularization error dynamics} is globally attractive if and only if \eqref{eqn: WPE condition} holds.\footnote{This implication is also briefly discussed in the paragraph before Proposition 1 of \cite{bruce2021necessary}.}
    Finally, the zero equilibrium of \eqref{eqn: time-varying regularization error dynamics} is globally attractive if and only of \eqref{eqn: time-varying regularization error dynamics} is globally attractive to zero.
\end{proof}

Next, Theorem \ref{theo: time varying regularization RLS global attractivity} concerns global attractivity of the error dynamics of RLS with time-varying regularization. 
It is shown that if the regularization term is designed to converge to zero, then the much weaker condition \eqref{eqn: RLS variable regulariztion rank condition} is sufficient for global attractivity.
Hence, weakly persistent excitation is not necessary for the error dynamics of RLS with time-varying regularization to be globally attractive to zero.
Furthermore, if the regularization term is designed to eventually be equal to zero, then the error dynamics of RLS with time-varying regularization are globally finite-time attractive to zero.

\begin{theo}
\label{theo: time varying regularization RLS global attractivity}
    Consider the assumptions and notation of Proposition \ref{prop: time-varying regularization error dynamics}.
    If $\lim_{k \rightarrow \infty} R_k = 0_{n \times n}$ and $(\Vert \theta_{{\rm reg},k} \Vert ) _{k=0}^\infty$ is bounded, and there exists $k_{\rm rank} \ge 0$ such that
    \begin{align}
    \label{eqn: RLS variable regulariztion rank condition}
        \eigmin\left[ \sum_{i=0}^{k_{\rm rank}} \phi_i^\rmT \Gamma_i \phi_i \right] > 0,
    \end{align}
    then \eqref{eqn: time-varying regularization error dynamics} is globally attractive to zero.
    If additionally, there exists $k_{\rm cut} \ge k_{\rm rank}$ such that, for all $k \ge k_{\rm cut}$, $R_k = 0_{n \times n}$, then \eqref{eqn: time-varying regularization error dynamics} is globally finite-time attractive to zero.
    
\end{theo}

\begin{proof}
It follows from \eqref{eqn: time-varying regularization batch least squares} and \eqref{eqn: linear measurement process no noise} that, for all $k \ge 1$,
\begin{align*}
    \theta_{k+1} &= (R_k + S_k)^{-1}( R_k \theta_{{\rm reg},k} + S_k \theta)
    \\
    &= (R_k + S_k)^{-1}( R_k \theta_{{\rm reg},k} - R_k \theta + R_k \theta + S_k \theta)
    \\
    &= (R_k + S_k)^{-1}[ R_k (\theta_{{\rm reg},k} - \theta) + ( R_k + S_k) \theta]
    \\
    &= \theta + (R_k + S_k)^{-1} R_k (\theta_{{\rm reg},k} - \theta),
\end{align*}
where $S_k \triangleq \sum_{i=0}^{k} \phi_i^\rmT \Gamma_i \phi_i$ for brevity.
Subtracting $\theta$ from both sides, it then follows from \eqref{eqn: parameter estimaton error defn} that, for all $k \ge 1$,
\begin{align}
\label{eqn: fading regularization error dynamics temp 1}
    \tilde{\theta}_{k+1} = (R_k + \sum_{i=0}^k \phi_i^\rmT \Gamma_i \phi_i)^{-1} R_k (\theta_{{\rm reg},k} - \theta)
\end{align}
For all $k \ge 1$, since $R_k$ and $R_k + \sum_{i=0}^k \phi_i^\rmT \Gamma_i \phi_i$ are positive semidefinite, it follows that
\begin{align}
\label{eqn: variable regularization global attractivity temp 1}
    \Vert \tilde{\theta}_{k+1} \Vert \le  
    \frac{\eigmax(R_k)}{\eigmin\left( R_k + \sum_{i=0}^k \phi_i^\rmT \Gamma_i \phi_i \right)} 
    \Vert \theta_{{\rm reg},k} - \theta \Vert.
\end{align}
Next, it follows from triangle inequality that, for all $k \ge 1$,
\begin{align}
\label{eqn: variable regularization global attractivity temp 2}
    \Vert \theta_{{\rm reg},k} - \theta \Vert \le \Vert \theta_{{\rm reg},k} \Vert + \Vert \theta \Vert.
\end{align}
Furthermore, it follows from \eqref{eqn: RLS variable regulariztion rank condition} that, for all $k \ge k_{\rm rank}$,
\begin{align}
\label{eqn: variable regularization global attractivity temp 3}
    \eigmin\big( R_k + \sum_{i=0}^k \phi_i^\rmT \Gamma_i \phi_i \big) \ge  \eigmin\big( \sum_{i=0}^{k_{\rm rank}} \phi_i^\rmT \Gamma_i \phi_i \big) > 0.
\end{align}
Substituting \eqref{eqn: variable regularization global attractivity temp 2} and \eqref{eqn: variable regularization global attractivity temp 3} into \eqref{eqn: variable regularization global attractivity temp 1}, it follows that, for all $k \ge k_{\rm rank}$,
\begin{align}
\label{eqn: variable regularization global attractivity temp 4}
    \Vert \tilde{\theta}_{k+1} \Vert \le  
    \frac{\eigmax(R_k)}{\eigmin\left( \sum_{i=0}^{k_{\rm rank}} \phi_i^\rmT \Gamma_i \phi_i  \right)} 
    (\Vert \theta_{{\rm reg},k} \Vert + \Vert \theta \Vert).
\end{align}
Finally, since $(\Vert \theta_{{\rm reg},k} \Vert ) _{k=0}^\infty$ is bounded and $\lim_{k \rightarrow \infty} R_k = 0_{n \times n}$ implies that $\lim_{k \rightarrow \infty} \eigmax(R_k) = 0$, it follows from \eqref{eqn: variable regularization global attractivity temp 4} that $\lim_{k \rightarrow \infty} \Vert \tilde{\theta}_{k} \Vert = 0$ and hence $\lim_{k \rightarrow \infty} \tilde{\theta}_{k} = 0_{n \times 1}$.
Thus, \eqref{eqn: time-varying regularization error dynamics} is globally attractive to zero.

Next, suppose there exists $k_{\rm cut} \ge k_{\rm rank}$ such that, for all $k \ge k_{\rm cut}$, $R_k = 0_{n \times n}$.
It follows from \eqref{eqn: RLS variable regulariztion rank condition} that, for all $k \ge k_{\rm cut}$, 
\begin{align*}
    \sum_{i=0}^k \phi_i^\rmT \Gamma_i \phi_i \succeq \sum_{i=0}^{k_{\rm rank}} \phi_i^\rmT \Gamma_i \phi_i \succ 0_{n \times n},
\end{align*}
and hence $\sum_{i=0}^k \phi_i^\rmT \Gamma_i \phi_i$ is nonsingular. 
Therefore, it follows from \eqref{eqn: fading regularization error dynamics temp 1} that, for all $k \ge k_{\rm cut}$, 
\begin{align*}
    \tilde{\theta}_{k+1} = (\sum_{i=0}^k \phi_i^\rmT \Gamma_i \phi_i)^{-1} 0_{n \times n} (\theta_{{\rm reg},k} - \theta) = 0_{n \times 1}.
\end{align*}
Thus \eqref{eqn: time-varying regularization error dynamics} is globally finite-time attractive to zero.
\end{proof}

\section{Fading-Regularization (FR) RLS}
This section proposes \textit{fading-regularization RLS} (FR-RLS), a simple algorithm that satisfies the conditions of Theorem \ref{theo: time varying regularization RLS global attractivity} sufficient for global finite-time attractivity to zero.
To begin, let $k_{\rm cut} \in \BBN_0 \cup \{\infty\}$, let $\mu \in (0,1)$, let $R_0 \in \BBR^{n \times n}$ be positive definite and, for all $k \ge 1$, define
\begin{align}
    R_k \triangleq \begin{cases}
        \mu^k R_0, & k < k_{\rm cut},
        \\
        0_{n \times n}, & k \ge k_{\rm cut}.
    \end{cases}
    \label{eqn: fading regularization Rk defn}
\end{align}
Similarly to classical RLS, the user selects a positive-definite regularization matrix $R_0$. 
In fading-regularization RLS, the user additional selects the \textit{fading factor} $\mu$ which controls how quickly regularization decays, and the \textit{regularization cutoff step} $k_{\rm cut}$, after which regularization is set to zero.

It follows from \eqref{eqn: variable regularization Pk update} that, for all $k \ge 0$, $P_{k+1}^{-1}$ is given by
\begin{align}
    & P_{k+1}^{-1} = \nonumber 
    \\
    & \begin{cases}
        R_k + \phi_k^\rmT \Gamma_k \phi_k, & k = 0,
        \\
        P_{k}^{-1} + \phi_k^\rmT \Gamma_k \phi_k - \mu^{k-1}(1-\mu)R_0, & 0 < k < k_{\rm cut},
        \\
        P_{k}^{-1} + \phi_k^\rmT \Gamma_k \phi_k - \mu^{k-1} R_0, & k = k_{\rm cut},
        \\
        P_{k}^{-1} + \phi_k^\rmT \Gamma_k \phi_k, & k > k_{\rm cut}.
    \end{cases}
\end{align}
Moreover, $\theta_1$ is still computed using \eqref{eqn: variable regularization theta1 update} and, for all $k \ge 1$, $\theta_{k+1}$ is computed using \eqref{eqn: variable regularization thetak update}.
Note, however, that to compute $\theta_{k+1}$ by \eqref{eqn: variable regularization thetak update}, it is necessary to first compute $P_{k+1}$. 
%
%
%
Computing $P_{k+1}$ by directly inverting $P_{k+1}^{-1}$ is an $\mathcal{O}(n^3)$ operation, which may be computationally expensive if $n \gg p$. It may be computationally faster in practice to compute the product $P_{k+1}[\cdots]$ in \eqref{eqn: variable regularization thetak update} using a Cholesky solver or LU solver instead of direct inversion of $P_{k+1}^{-1}$. However, this is still an $\mathcal{O}(n^3)$ operation.

Directly inverting $P_{k+1}^{-1}$ can be avoided for $k > k_{\rm cut}$ by using the matrix inversion lemma, given by Lemma \ref{lem: matrix inversion lemma}.
It follows from Lemma \ref{lem: matrix inversion lemma} for all $k \ge k_{\rm cut}$, $P_{k+1}$ can be expressed as
\begin{align}
\label{eqn: RLS Pk update}
    P_{k+1} = P_k - P_k \phi_k^\rmT(\Gamma_k^{-1} + \phi_k P_k \phi_k^\rmT)^{-1} \phi_k P_k.
\end{align}
Note that this is the same as the classical RLS update equation for $P_{k+1}$ since there is no time-varying regularization for $k > k_{\rm cut}$.
Using \eqref{eqn: variable regularization thetak update} and \eqref{eqn: RLS Pk update}, it follows that, for all $k > k_{\rm cut}$, $\theta_{k+1}$ can be computed recursively in $\mathcal{O}(p n^2)$ complexity.
However, for all $0 < k \le k_{\rm cut}$, matrix inversion lemma cannot improve computational complexity to compute $P_{k+1}$ since $\rank(R_{k} - R_{k-1}) = n$.
Hence, for all $0 < k \le k_{\rm cut}$, $\theta_{k+1}$ must be computed in $\mathcal{O}(n^3)$ complexity, which may be significantly slower than $\mathcal{O}(p n^2)$ if $n \gg p$.

\section{Rank-1 Fading-Regularization (R1FR) RLS}
The $\mathcal{O}(n^3)$ computational complexity of FR-RLS motivates \textit{rank-1 fading-regularization RLS} (R1FR-RLS), an algorithm that satisfies the conditions of Theorem \ref{theo: time varying regularization RLS global attractivity} sufficient for global finite-time attractivity to zero, but which runs in the same time complexity as classical RLS.
The key idea behind this algorithm is to design the regularization matrix such that, for all $k \ge 1$, $\rank(R_k - R_{k-1}) \le 1$.

To begin, let $j_{\rm cut} \in \BBN_0 \cup \{\infty\}$, let $\mu \in (0,1)$, and let $R_0 \in \BBR^{n \times n}$ be positive definite.
Let $d_{0,1},\hdots,d_{0,n} \in \BBR$ and $v_{0,1},\hdots,v_{0,n} \in \BBR^{n \times 1}$ respectively be eigenvalues and eigenvectors of $R_0$ satisfying
\begin{align}
    R_0 = d_{0,1} v_{0,1} v_{0,1}^\rmT + \hdots + d_{0,n} v_{0,n} v_{0,n}^\rmT.
\end{align}
Next, for all $j \ge 0$ and for all $l = 0, \hdots, n-1$, define
\begin{align}
\label{eqn: Rank 1 fading regularization Rk defn}
    & R_{jn + l} \triangleq \nonumber
    \\
    & 
    \begin{cases}
        \mu^{jn} ( \mu \sum_{i=1}^{l} d_{0,i} v_{0,i} v_{0,i}^\rmT  + \sum_{i=l+1}^{n} d_{0,i} v_{0,i} v_{0,i}^\rmT ), & j < j_{\rm cut},
        \\
        \mu^{jn} \sum_{i=l+1}^{n} d_{0,i} v_{0,i} v_{0,i}^\rmT, & j = j_{\rm cut},
        \\
        0_{n \times n}, & j > j_{\rm cut}.
    \end{cases}    
\end{align}
For all $k \ge 0$, since $R_k$ is a linear combination of $\{v_{0,i} v_{0,i}^\rmT\}_{i=1}^n$, it follows that $R_k$ is positive semidefinite. 
Furthermore, note that, for all $0 \le j \le j_{\rm cut}$, 
\begin{align}
    R_{jn} = \mu^{jn} R_0.
\end{align}
Hence, the regularization term in rank-1 fading-regularization is the same as that of fading regularization every $n$ steps.
Next, note that, for all $k \ge 1$, $R_k$ can be computed recursively as
\begin{align}
    R_{k} = \begin{cases}
        R_{k-1} - c_k v_{0,l} v_{0,l}^\rmT, & k \le (j_{\rm cut}+1)n,
        \\
        0_{n \times n}, & k > (j_{\rm cut}+1)n,
    \end{cases}
    \label{eqn: rank 1 fading regularization Rk update}
\end{align}
where
\begin{align}
        c_k &= \begin{cases}
        \mu^{jn} (1-\mu^n) d_{0,l}, & 0 < k \le Mn,
        \\
        \mu^{jn} d_{0,l}, & Mn < k \le (j_{\rm cut}+1) n,
    \end{cases}
\end{align}
and where $j \triangleq \lfloor \frac{k-1}{n} \rfloor$ and $l \triangleq (k-1) \Mod n$.
It is evident from \eqref{eqn: rank 1 fading regularization Rk update} that, for all $k \ge 1$, $R_k - R_{k-1}$ is a scalar multiple of $v_{0,l} v_{0,l}^\rmT$. 
Thus, for all $k \ge 1$, $\rank(R_k - R_{k-1}) \le 1$.
Next, it follows from \eqref{eqn: variable regularization P1 update}, \eqref{eqn: variable regularization Pk update}, and \eqref{eqn: rank 1 fading regularization Rk update} that, for all $k \ge 0$, $P_{k+1}^{-1}$ is given by
\begin{align}
    & P_{k+1}^{-1} = \nonumber
    \\
    & \begin{cases}
        R_0 + \phi_0^\rmT \Gamma_0 \phi_0,   &    k = 0,
        \\
        P_{k}^{-1} + \phi_k^\rmT \Gamma_k \phi_k - c_k v_{0,l} v_{0,l}^\rmT, & 0 < k \le (j_{\rm cut}+1) n,
        \\
         P_{k}^{-1} + \phi_k^\rmT \Gamma_k \phi_k, & k > (j_{\rm cut}+1)n.
    \end{cases}
\end{align}
Furthermore, note that, for all $0 < k \le (j_{\rm cut}+1)n$,
\begin{align}
    \phi_k^\rmT \Gamma_k \phi_k - c_k v_{0,l} v_{0,l}^\rmT 
    = 
    \begin{bmatrix}
        \phi_k^\rmT
        &
        v_{0,l}
    \end{bmatrix}
    \begin{bmatrix}
        \Gamma_k & 0_{p \times 1}
        \\
        0_{1 \times p} & -c_k
    \end{bmatrix}
    \begin{bmatrix}
        \phi_k
        \\
        v_{0,l}^\rmT
    \end{bmatrix}.
\end{align}
Thus, it follows from Lemma \ref{lem: matrix inversion lemma} that, for all $k \ge 0$, $P_{k+1}$ is given by
\begin{align}
    P_{k+1} = P_k - P_k \bar{\phi}_k^\rmT (\bar{\Gamma}_k^{-1} + \bar{\phi}_k P_k \bar{\phi}_k^\rmT)^{-1}\bar{\phi}_k P_k,
    \label{eqn: rank 1 fading regularization Pk update}
\end{align}
where $\bar{\phi}_k \in \BBR^{p \times n} \cup \BBR^{(p+1) \times n}$ is defined
\begin{align}
    \bar{\phi}_k \triangleq 
    \begin{cases}
        \phi_k, & k = 0,
        \\
        \begin{bmatrix}
            \phi_k
            \\
            v_{0,l}^\rmT
        \end{bmatrix},
        & 
        0 < k \le (j_{\rm cut}+1)n,
        \\
        \phi_k, & k \ge (j_{\rm cut}+1)n,
    \end{cases}
\end{align}
where $\bar{\Gamma}_k \in \BBR^{p \times p} \cup \BBR^{(p+1) \times {p+1}}$ is defined
\begin{align}
    \bar{\Gamma}_k \triangleq 
    \begin{cases}
    \Gamma_k, & k = 0,
    \\
        \begin{bmatrix}
            \Gamma_k, & 0_{p \times 1}
            \\
            0_{1 \times p} & -c_k
        \end{bmatrix},
        & 
        0 < k \le (j_{\rm cut}+1)n,
        \\
        \Gamma_k, & k \ge (j_{\rm cut}+1)n,
    \end{cases}
\end{align}
and where $P_0 \triangleq R_0^{-1}$.
The computational complexity to compute $P_{k+1}$ using \eqref{eqn: RLS Pk update} is $\mathcal{O}((p+1)n^2)$ when $0 < k \le (j_{\rm cut}+1)n$ and $\mathcal{O}(p n^2)$ otherwise.
Therefore, for all $k \ge 0$, using \eqref{eqn: variable regularization theta1 update}, \eqref{eqn: variable regularization thetak update}, and \eqref{eqn: rank 1 fading regularization Pk update}, $\theta_{k+1}$ can be computed recursively in at most $\mathcal{O}((p+1)n^2)$ complexity.

\section{Numerical Examples}

\begin{example}[No Measurement Noise.]
\label{example: time-varying regularization no measurement noise}
    Consider $n = 100$ parameters and $p = 2$ measurements per step.
    The true parameters $\theta \in \BBR^{100}$ are sampled from $\mathcal{N}(0_{100 \times 1},I_{100})$.
    Let $R_0 = r_0 I_{100}$ with $r_0 = 1$ and, for all $k \ge 0$, let $\theta_{{\rm reg},k} = 0_{100 \times 1}$.
    We consider three choices of regularization:
    \begin{enumerate}
        \item For all $k \ge 0$, $R_k = R_0$. We label this {RLS}.
        \item For all $k \ge 0$, $R_k$ is given by \eqref{eqn: fading regularization Rk defn} with $\mu = 0.99$ and $k_{\rm cut} = 201$. We label this {FR-RLS}.
        \item For all $k \ge 0$, $R_k$ is given by \eqref{eqn: Rank 1 fading regularization Rk defn} with $\mu = 0.99$ and $j_{\rm cut} = 1$. We label this {R1FR-RLS}.
    \end{enumerate}
    We choose $k_{\rm cut} = 201$ and $j_{\rm cut} = 1$ for easy comparison of FR-RLS and R1FR-RLS as in both cases, $R_k = 0_{100 \times 100}$ for all $k \ge 201$.

    Next, for all $k \ge 0$, let $\Gamma_k = I_p$.
    We consider two cases of regressors. First, consider that, for all $k \ge 0$, the rows of $\phi_k \in \BBR^{2 \times 100}$ are i.i.d. sampled from $\mathcal{N}(0_{100 \times 1},I_{100})$.
    We label this \textit{PE Data} as the regressors are persistently exciting.
    Next, consider that, for all $0 \le k \le 100$, the rows of $\phi_k \in \BBR^{2 \times 100}$ are i.i.d. sampled from $\mathcal{N}(0_{100 \times 1},I_{100})$, while, for all $k > 100$, $\phi_k = 0_{2 \times 100}$. 
    We label this \textit{Non-PE Data}.
    Finally, for all $k \ge 0$, let $y_k = \phi_k \theta$ and let $\theta_{k+1}$ be given by \eqref{eqn: variable regularization theta1 update} and \eqref{eqn: variable regularization thetak update}.
    The norm of parameter estimation error, $\Vert \theta_k - \theta \Vert$, is plotted in Figure \ref{fig: time varying regularization no noise} for these six cases.
    \begin{figure}[ht]
        \centering
        \includegraphics[width=\linewidth]{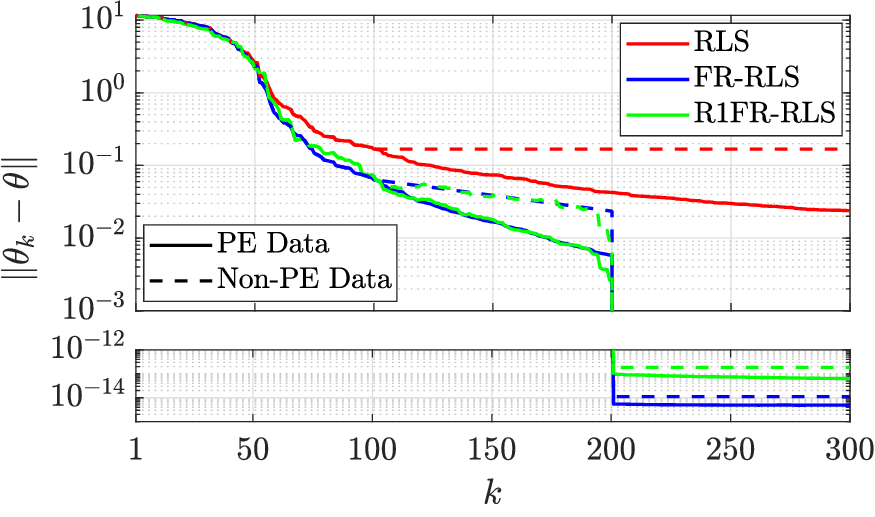}
        \caption{Example 1: Parameter estimation error, $\Vert \theta_k - \theta \Vert$ versus time step $k$ for RLS, fading regularization RLS, and rank-1 fading regularization RLS. 
        Solid lines indicate that data is persistently exciting, while dashed lines indicates that data is not persistently exciting for $k \ge 100$.
        }
        \label{fig: time varying regularization no noise}
    \end{figure}
    
    Note that in RLS, parameter estimation error does converge to zero in the non-PE case, and only converges asymptotically in the PE case. This is indicative of regularization-induced bias.
    In FR-RLS and R1FR-RLS, in both the PE and non-PE case, parameter estimation error converges to zero, within numerical precision.
    Moreover, convergence happens in finite time at step $k = 201$, due to cutoff of regularization.
    This verifies the results of Theorem \ref{theo: time varying regularization RLS global attractivity}.
    $\hfill\mbox{$\diamond$}$
\end{example}
\begin{example}[White Measurement Noise.]
    Consider the same setup as the PE Data case of Example \ref{example: time-varying regularization no measurement noise}.
    However, for all $k \ge 0$, we let $y_k = \phi_k \theta + w_k$, where $w_k$ is i.i.d. sampled from $\mathcal{N}(0_{p \times 1},I_p)$.
    While many methods exist to choose a suitable regularization term based on properties of the measurement noise, e.g. \cite{hoerl1970ridge,bauer2007regularization,ballal2017bounded}, we will assume properties of the measurement noise are not well known.
    We consider the three values $r_0 = 0.01$, $r_0 = 1$, and $r_0 = 100$, and initialize $R_0 = r_0 I_{100}$.
    
    For each value of $r_0$, we run the RLS, FR-RLS, and R1FR-RLS algorithms, as detailed in Example \ref{example: time-varying regularization no measurement noise}.
    Each algorithm is run for 1000 independent trials. 
    Figure \ref{fig: time-varying regularization with noise} plots the mean (line) and 95\% confidence interval (shaded) of $\Vert \theta_k - \theta \Vert$ over the 1000 trials.
    \begin{figure}[ht]
        \centering
        \includegraphics[width=\linewidth]{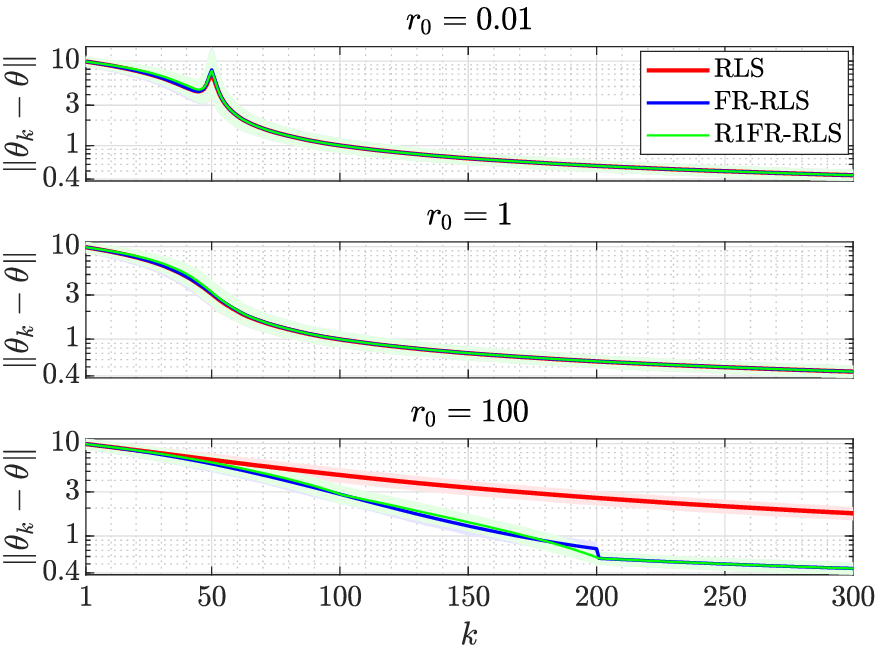}
        \caption{Exmaple 2: Parameter estimation error, $\Vert \theta_k - \theta \Vert$ versus time step $k$ for RLS, fading regularization, and rank-1 fading regularization, over 1000 independent trials.
        Solid line shows the mean of the 1000 trials and shaded shows the 95\% confidence interval.
        }
        \label{fig: time-varying regularization with noise}
    \end{figure}

    In this example, $r_0 = 0.01$ results in under-regularization, creating a sensitivity to measurement noise at $k = 49$, where $\sum_{i=0}^{k} \phi_i^\rmT \Gamma_i \phi_i$ attains full rank.
    $r_0 = 1$ gives a good amount of regularization while $r_0 = 100$ results in over-regularization.
    In the $r_0 = 0.01$ and $r_0 = 1$ cases, RLS, FR-RLS, and R1FR-RLS all perform nearly identically. 
    Since bias due to regularization is relatively small in these cases, adding a time-varying regularization term has little impact on performance. 
    In the $r_0 = 100$ case, there is significant regularization-induced bias due to a large regularization term.
    This results in slow identification in RLS.
    FR-RLS and R1FR-RLS, however, gradually reduce this bias during $1 \le k \le 200$, and remove this bias altogether at $k = 201$. 
    
    Therefore, FR-RLS and R1FR-RLS have minimal impact on performance in the case of under-regularization or well-tuned regularization, while gradually removing bias in the case of over-regularization.
    Hence, these time-varying regularization schemes can be added to existing RLS identification to protect against over-regularization when properties of measurement noise are not well known.

    \begin{figure}[ht]
        \centering
        \includegraphics[width=\linewidth]{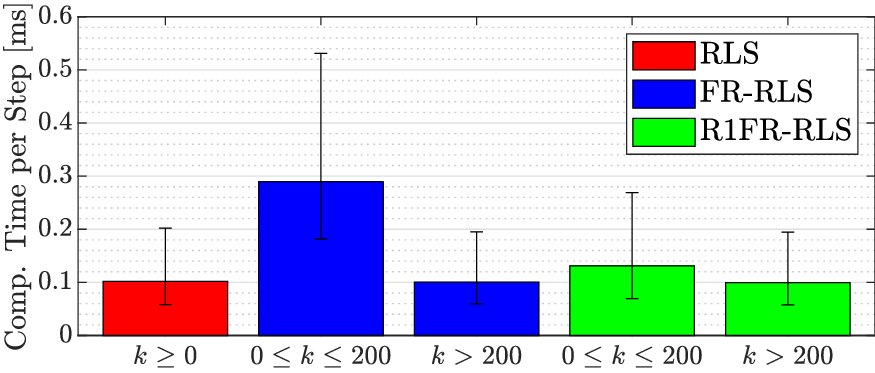}
        \caption{Example 2: Computation time per step in milliseconds of RLS, fading regularization, and rank-1 fading regularization.
        Fading regularization and rank-1 fading regularization are divided into $0 \le k \le 200$ (time-varying regularization), and $k > 200$ (constant zero regularization).
        Error bars give the 95\% confidence intervals.
        }
        \label{fig: time-varying regularization computation time}
    \end{figure}
    Finally, Figure \ref{fig: time-varying regularization computation time} shows the computation time per step of RLS, FR-RLS and R1FR-RLS.
    Since FR-RLS and R1FR-RLS have constant regularization for $k > 200$, we see very similar computation time per step between RLS, FR-RLS for $k > 200$, and R1FR-RLS for $k > 200$. 
    The computation time per step of FR-RLS during $0 \le k \le 200$ is approximately three times that of RLS, a result of $\mathcal{O}(n^3)$ complexity per step.
    This difference in computation time is further exaggerated as $n \gg p$.
    R1FR-RLS reduces computation cost during $0 \le k \le 200$ to comparable to RLS, a result of $\mathcal{O}((p+1)n^2)$ complexity using matrix inversion lemma.
    $\hfill\mbox{$\diamond$}$
\end{example}
\section{Conclusion}
This work presented RLS with time-varying regularization, sufficient conditions for global finite-time attractivity of parameter estimation error to zero, and two algorithms that satisfy these sufficient conditions: fading regularization RLS and rank-1 fading regularization RLS.
Numerical examples validate these theoretical guarantees and show how fading regularization can protect against over-regularization.
While this is one potential use case, this work also opens up the potential to new time-varying regularization designs. 
Another area of future interest is augmenting time-varying regularization with the existing literature of RLS forgetting algorithms.
%


\bibliographystyle{IEEEtran}
\bibliography{refs}

\appendix
\begin{lema}{A.1}
\label{lem: quadratic cost minimizer}
Let $A \in \BBR^{n \times n}$ be positive definite, let $b \in \BBR^n$ and $c \in \BBR$, and define $f\colon\BBR^n \rightarrow \BBR$ by $ f(x) \triangleq x^\rmT A x + 2 b^\rmT x + c$.
%
%
Then, $f$ has a unique global minimizer given by $\argmin_{x \in \BBR^n} f(x) = -A^{-1} b.$
\end{lema}
\begin{lema}{A.2}[Matrix Inversion Lemma]
\label{lem: matrix inversion lemma}
Let $A \in \BBR^{n \times n}$, $U \in \BBR^{n \times p}$, $C \in \BBR^{p \times p}$, $V \in \BBR^{p \times n}$. Assume $A$, $C$, and $A+UCV$ are nonsingular. Then, $(A+UCV)^{-1} = A^{-1} - A^{-1}U(C^{-1} + VA^{-1} U)^{-1} V A^{-1}$.
\end{lema}

\end{document}